\documentclass[a4paper,11pt]{article}

\usepackage{array}
\usepackage{amsthm}
\usepackage{amssymb}
\usepackage{amsmath}
\usepackage{appendix}
\usepackage{bbold}
\usepackage{bbm}
\usepackage{cancel}
\usepackage{color}
\usepackage{colortbl} 
\usepackage{ctable}
\usepackage{enumerate}
\usepackage{fullpage}
\usepackage{graphicx,psfrag,epsf}
\usepackage[latin1]{inputenc}
\usepackage{listings}
\usepackage{lscape}
\usepackage{mathdots}
\usepackage{multirow}
\usepackage{natbib}
\usepackage{setspace}
\usepackage{slashbox}
\usepackage{subfigure} 
\usepackage{url}

\def\bx{\boldsymbol{x}}

\def\bU{\boldsymbol{U}}

\def\bA{\boldsymbol{A}}

\def\bL{\boldsymbol{L}}
\def\bP{\boldsymbol{P}}
\def\bQ{\boldsymbol{Q}}
\def\bR{\boldsymbol{R}}
\def\bS{\boldsymbol{S}}
\def\bmu{\boldsymbol{\mu}}

\def\bSigma{\boldsymbol{\Sigma}}

\def\bLambda{\boldsymbol{\Lambda}}
\def\bPsi{\boldsymbol{\Psi}}

\newcommand{\real}{I\hspace{-1.0mm}R}

\newcommand{\diag}{\mathop{\mathrm{diag}}}

\def\bSig\mathbf{\Sigma}

\newtheorem{definition}{Definition}[section]
\newtheorem{theorem}{Theorem}[section]

\begin{document}

\title{Unconstrained representation of orthogonal matrices with application to common principle components}

\author{
Luca Bagnato\thanks{
Universit\`{a} Cattolica del Sacro Cuore, 
Dipartimento di Scienze Economiche e Sociali, 
Via Emilia Parmense, 29122 Piacenza, Italy, e.mail: \texttt{luca.bagnato@unicatt.it}}
\quad  
Antonio Punzo\thanks{
Universit\`{a} di Catania,
Dipartimento di Economia e Impresa,
Corso Italia 55, 95129, Catania, Italy, e.mail: \texttt{antonio.punzo@unict.it}}
}

\date{}

\maketitle

\begin{abstract}
Many statistical problems involve the estimation of a $\left(d\times d\right)$ orthogonal matrix $\bQ$.
Such an estimation is often challenging due to the orthonormality constraints on $\bQ$.
To cope with this problem, we propose a very simple decomposition for orthogonal matrices which we abbreviate as PLR decomposition.
It produces a one-to-one correspondence between $\bQ$ and a $\left(d\times d\right)$ unit lower triangular matrix $\bL$ whose $d\left(d-1\right)/2$ entries below the diagonal are unconstrained real values.
Once the decomposition is applied, regardless of the objective function under consideration, we can use any classical unconstrained optimization method to find the minimum (or maximum) of the objective function with respect to $\bL$.
For illustrative purposes, we apply the PLR decomposition in common principle components analysis (CPCA) for the maximum likelihood estimation of the common orthogonal matrix when a multivariate leptokurtic-normal distribution is assumed in each group.
Compared to the commonly used normal distribution, the leptokurtic-normal has an additional parameter governing the excess kurtosis; this makes the estimation of $\bQ$ in CPCA more robust against mild outliers.
The usefulness of the PLR decomposition in leptokurtic-normal CPCA is illustrated by two biometric data analyses.

\end{abstract}


\section{Introduction}
\label{sec:intro}

With the term orthogonal matrix we refer to a $\left(d \times d\right)$ matrix $\bQ$ whose columns are mutually orthogonal unit vectors (i.e., orthonormal vectors).
As highlighted by \citet[][p.~209]{Bane:Line:2014}, one might, perhaps more properly, call $\bQ$ an ``orthonormal'' matrix, but the more conventional name is an ``orthogonal'' matrix, and we will adopt it hereafter.
For further characterizations, properties, and details about orthogonal matrices see, e.g., \citet[][Chapter~9.10]{lutkepohl96}, \citet[][Chapter~3.5]{Heal:Matr:2000}, \citet[][Chapter~1.10]{Scho:Matr:2016}, and \citet[][Chapter~5.4]{Sear:Khur:Matr:2017}.  
Orthogonal matrices are used extensively in statistics, especially in linear models and multivariate analysis (see, e.g., \citealp[][Chapter~11]{Gray:AnIn:1976} and \citealp{Jame:Norm:1954}).

The $d^2$ elements of $\bQ$ are subject to $d\left(d+1\right)/2$ (orthonormality) constraints.
It is therefore not surprising that they can be represented by only $d^2-d\left(d+1\right)/2=d\left(d-1\right)/2$ independent parameters.
A representation is convenient if $\bQ$ can be quickly computed from these $d\left(d-1\right)/2$ independent parameters.
Such a representation should facilitate the search for an orthogonal matrix that satisfies a certain optimality criterion induced by the chosen estimation method, especially if these independent parameters were real-valued \citep{Khur:Adva:2003}.
Methods to parameterize an orthogonal matrix are reviewed in \citet{Khur:Andr:TheP:1989}.
A similar problem is bumped into when a $\left(d \times d\right)$ positive-definite matrix $\bSigma$, often encountered in statistics in the form of a covariance matrix, needs to be estimated.
Luckily, in this case, the Cholesky decomposition allows to map the $d\left(d+1\right)/2$ independent parameters of $\bSigma$ with the $d\left(d+1\right)/2$ real-valued elements of a $\left(d \times d\right)$ lower triangular matrix (\citealp{pourahmadi99,pourahmadi00} and \citealp{pourahmadi07}).

Unfortunately, an analogous of the Cholesky decomposition does not exist for $\bQ$.
In this paper, we fill the gap by providing the PLR decomposition: similarly to the Cholesky decomposition, it maps $\bQ$ to a $\left(d \times d\right)$ unit lower triangular matrix $\bL$ whose $d\left(d-1\right)/2$ entries below the diagonal are real-valued. 
The PLR decomposition is based on the famous QR and PLU factorizations which are used, in our context, for squared and invertible matrices.

For illustrative purposes, we apply the PLR decomposition in common principal component analysis (CPCA), where the space spanned by the $d$ vectors (principal components) of $\bQ$ is assumed to be identical across several known groups, whereas the variances associated with the common principal components may vary.
When the groups are assumed to be normally distributed, as typically happens, we can use the FG algorithm developed by \citet{Flur:Gaut:anal:1986} for the estimation of $\bQ$.
Although the FG algorithm is distribution-free \citep[][p.~71]{Flur:comm:1988}, under non-normal distributions it may provide an orthogonal matrix which does not maximize the likelihood.
Motivated by this consideration, we assume groups having a leptokurtic-normal distribution and use the PLR decomposition to allow $\bQ$ to be estimated by any standard unconstrained maximization routine.
The leptokurtic-normal is an heavy-tailed generalization of the normal distribution that can be preferred, for example, in the presence of mild outliers.

The paper is organized as follow.
In Section~\ref{sec:Decomposition of orthogonal matrices} we introduce the PLR decomposition.
In Section~\ref{sec:Application} we first define the CPCA based on leptokurtic-normal groups, and then consider the PLR decomposition in the estimation of the common orthogonal matrix.
In Section~\ref{sec:Application to allometric studies} we illustrate the leptokurtic-normal CPCA in allometric studies by using two well-known biometric data sets, where the method shows its better performance with respect to the consolidated normal CPCA.
Nevertheless, we want to stress that our goal is not to propose a new robust method for estimating the common principal components.
Instead, we simply want to illustrate how, regardless of the way the orthogonal matrix enters in the considered model, we can use our PLR decomposition, along with any unconstrained optimization routine, to estimate $\bQ$, without the need to define \textit{ad-hoc} estimating algorithms.
Finally, we give conclusions and avenues for further research in Section~\ref{sec:Conclusions}.
 
\section{PLR decomposition of orthogonal matrices}
\label{sec:Decomposition of orthogonal matrices}

Before to present the PLR decomposition for orthogonal matrices, Definitions~\ref{theo:2pre} and \ref{theo:1pre} recall the well-known QR and PLU decompositions.
\begin{definition}[QR decomposition]\label{theo:2pre} 
If $\bA$ is a $\left(d\times d\right)$ invertible matrix, then there is a unique $\left(d\times d\right)$ orthogonal matrix $\bQ$, and a unique $\left(d\times d\right)$ upper triangular matrix $\bR$, such that
$\bA = \bQ\bR$.
\end{definition}
\begin{definition}[PLU decomposition]\label{theo:1pre} 
If $\bA$ is a $(d\times d)$ invertible matrix, then there is a $(d\times d)$ permutation matrix $\bP$, a $\left(d\times d\right)$ unit lower triangular matrix $\bL$, and a $\left(d\times d\right)$ upper triangular matrix $\bU$, such that
$\bA = \bP\bL\bU$.
\end{definition}
The following theorem presents our PLR decomposition.
\begin{theorem}[PLR decomposition]\label{theo:1dec}
If $\bQ$ is a $\left(d\times d\right)$ orthogonal matrix, then it can be factorized as
\begin{equation}
\bQ =  \bP \bL \bR^{-1},
\label{eq:decPLR}
\end{equation}
where $\bP$ and $\bL$ are defined as in Definition~\ref{theo:1pre}, and $\bR$ is the upper triangular matrix coming from the QR decomposition of $\bA=\bP\bL$ (see Definition~\ref{theo:2pre}).
\end{theorem}
\begin{proof}
Because $\bQ$ is an orthogonal matrix, it is invertible and, according to Definition~\ref{theo:1pre}, admits the PLU decomposition
\begin{equation}
\bQ = \bP \bL \bU.
\label{eq:decPLU}
\end{equation}
Because any unit lower triangular matrix is invertible, $\bL$, as well as $\bA=\bP\bL$, are invertible. 
Then, according to Definition~\ref{theo:2pre}, $\bA$ admits the QR decomposition $\bA=\bQ\bR$, which we recall to be unique.
Then, it is easy to verify that $\bU$ must be equal to $\bR^{-1}$, and the theorem is proved.
\end{proof}

$\bL$ is the key matrix of the PLR decomposition in \eqref{eq:decPLR}; indeed, $\bP$ can be thought to as a sort of nuisance matrix only affecting the ordering of the columns of $\bQ$ -- and we know that such an ordering is often not of interest -- and $\bR$ is a function of $\bL$.
In particular, $\bR$ is the upper triangular matrix coming from the QR decomposition of $\bA=\bP\bL$.
Note that, the number of free elements in $\bL$ (which are those below the main diagonal) is $d\left(d-1\right)/2$, as the number of free elements of the orthogonal matrix $\bQ$.
This means that: 1) any orthogonal matrix admits the PLR decomposition in \eqref{eq:decPLR}, and 2) any pair $\left\{\bP,\bL\right\}$ is associated to an orthogonal matrix.

\section{Leptokurtic-normal common principal components}
\label{sec:Application}

The advantages of our PLR decomposition can be appreciated in many statistical fields. 
Among them there is the common principal component analysis. 
Below, we give an example by considering groups being distributed according to a multivariate leptokurtic-normal.

\subsection{Preliminaries}
\label{subsec:Preliminaries}

Let $\left\{\bx_{ij};\text{ $i=1,\ldots,n_j$, $j=1,\ldots,k$}\right\}$, with $\bx_{ij}\in\real^d$, be a set of $n=\sum_{j=1}^kn_j$ independent observations from $k$ independent groups (or subpopulations) having mean $\bmu_j$ and covariance matrix $\bSigma_j$.
If the inferential interest is on $\bSigma_1,\ldots,\bSigma_k$, then there is the need to estimate $kd\left(d+1\right)/2$ parameters.
Such a number may be excessive when $k$, but especially $d$, are large, causing problems in the estimation phase.
These problems can often be avoided if $\bSigma_1,\ldots,\bSigma_k$ exhibit some common structure, and several models have been proposed in this direction
(see, e.g., \citealp{Flur:comm:1984,Flur:Prop:1986,Flur:Two:1987}, \citealp{Boik:Spec:2002}, and \citealp{Gres:Ingr:Punz:Asse:2011}).
The assessment of a common covariance structure, in addition to allow for parsimony, can provide more information about the group conditional distributions \citep{Gres:Punz:Clos:2013} and it is of intrinsic interest in several fields such as biometry (refer to Section~\ref{sec:Application to allometric studies}). 

Most of the existing common covariance structures are based on the eigen-decomposition $\bSigma_j=\bQ_j\bLambda_j\bQ_j'$, $j=1,\ldots,k$, where $\bLambda_j=\diag\left(\lambda_{j1},\ldots,\lambda_{jd}\right)$ and $\bQ_j$ denote the eigenvalues and eigenvectors matrices, respectively. 
A famous common structure assumes that the $k$ covariance matrices have possibly different eigenvalues but identical eigenvectors, i.e.,
\begin{equation}
\bSigma_j=\bQ\bLambda_j\bQ', \quad j=1,\ldots,k.
\label{eq:N-CPC model}
\end{equation}
Model~\eqref{eq:N-CPC model} was proposed by \citet{Flur:comm:1984} and is known as the common principal components (CPC) model. 
\citet{Flur:comm:1984} also derived the maximum likelihood (ML) estimators of $\bQ$ and $\bLambda_j$, assuming $d$-variate normality in each group.
The asymptotic distribution of these estimators was studied by \citet{Flur:Asym:1986}.
The corresponding likelihood function can be written as
\begin{equation}
l_{\text{N-CPC}}\left(\bPsi\right) = C -\frac{1}{2} \sum_{j=1}^k n_j\left[ \sum_{h=1}^d \ln \lambda_{jh} + \mbox{tr}\left(\bQ \bLambda_j^{-1}\bQ' \bS_j\right)\right],
\label{eq:N-CPC loglik}
\end{equation}
where $\bPsi=\left\{\bmu_j,\bQ,\bLambda_j;\text{ $j=1,\ldots,k$}\right\}$ is the whole set of parameters of cardinality $m_{\text{N-CPC}}=kd + d\left(d-1\right)/2 + kd$, $C$ is a constant that does not depend on the parameters, and $\bS_j=n_j^{-1}\sum_{i=1}^{n_j}\bS_{ij}$, with $\bS_{ij}= \left(\bx_{ij}-\bmu_j\right)\left(\bx_{ij}-\bmu_j\right)'$.
A closed-form solution for the ML estimate of $\bQ$ does not exist, but the Flury-Gautschi (FG) algorithm of \citet{Flur:Gaut:anal:1986}, which is based on pairwise rotations of orthogonal vectors \citep{Flur:Cons:theF:1985}, can be used to obtain such a solution.
The more appropriate the CPC model is, the more able the ML-estimated common orthogonal matrix $\bQ$ is to simultaneously rotate the sample covariance matrices to nearly diagonal form.
\citet{Flur:comm:1984} also proposed a likelihood-ratio test having the CPC model under the null and the unconstrained model under the alternative.
The monograph by \citet{Flur:comm:1988} provides a rigorous treatise of the CPC and related models, detailing their properties and offering several examples, with a special focus on biometric data.

\subsection{The model}
\label{subsec:The model}

However, as confirmed by the simulations of \citet{Hall:Pain:Verd:Test:2010}, the CPC model discussed above is quite sensitive to the violation of group-specific multivariate normality (see, e.g., \citealp{Boen:Orel:Arob:2001} and \citealp{Boen:Pire:Rodr:Infl:2002,Boen:Pire:Rodr:Gene:2006,Boen:Pire:Rodr:Robu:2009} for examples of robust CPC approaches). 
These violations are often due to the presence of mild outliers. 
Contextualizing in CPCA the definition given by \citet[][p.~4 and pp.~79--80; see also \citealp{Mazz:Punz:Mixt:2017}]{Ritt:Robu:2015}, we can define as mild outlier in group $j$ a point that does not deviate from the reference multivariate normal distribution of that group and is not strongly outlying but, rather, it produces an overall group-specific distribution with heavier tails.
Therefore, to reduce the influence of these points, more-flexible elliptically symmetric heavy-tailed distributions can be considered \citep[][p.~4 and pp.~79]{Ritt:Robu:2015}.
Following this idea, we consider the multivariate leptokurtic-normal distribution \citep{Bagn:Punz:Zoia:Them:2016} in CPCA.
Compared to the normal distribution, the leptokurtic-normal has an additional parameter $\beta$ governing the excess kurtosis and, advantageously with respect to other heavy-tailed elliptical distributions, its parameters correspond to quantities of direct interest (mean, covariance matrix, and excess kurtosis).
Such a distribution was successfully applied in the modelling of biometric and financial data (\citealp{Bagn:Punz:Zoia:Them:2016} and \citealp{Maru:Punz:Bagn:Hidd:2019}). 

The probability density function (pdf) of a $d$-variate leptokurtic-normal (LN) distribution with mean vector $\bmu$, covariance matrix $\bSigma$, and excess kurtosis $\beta \in \left[0,\beta_{\max}\right]$, where $\beta_{\max}=\min\left[4d,4d\left(d+2\right)/5\right]$, is given by
\begin{equation}
f_{\text{LN}}\left(\bx;\bmu,\bSigma,\beta\right) = q\left[\left(\bx-\bmu\right)'\bSigma^{-1}\left(\bx-\bmu\right);\beta\right] f_{\text{N}}\left(\bx; \bmu,\bSigma\right),
\quad  \bx \in \real^d,
\label{eq:MLN}
\end{equation}
where $f_{\text{N}}\left(\cdot ; \bmu, \bSigma\right)$ is the pdf of a $d$-variate normal distribution with mean vector $\bmu$ and covariance matrix $\bSigma$, and 
	\begin{equation*}
q\left(y;\beta\right) = 1 + \frac{\beta}{8 d\left(d+2\right)}\left[y^2-2\left(d+2\right)y+d\left(d+2\right)\right], \quad y\geq 0.
\label{eq:moddenmm}
\end{equation*}
The constraint $\beta \in \left[0,\beta_{\max}\right]$ is the intersection of: i) $\beta \in \left[0,4d\right]$, which assures that the pdf is positive elliptical, and ii) $\beta \in \left[0,4d\left(d+2\right)/5\right]$, which guarantees that the pdf is unimodal.
As a special case, $f_{\text{LN}}\left(\bx;\bmu,\bSigma,\beta\right)$ coincides with $f_{\text{N}}\left(\bx;\bmu,\bSigma\right)$ for $\beta=0$. 

The log-likelihood function of the CPC model based on leptokurtic-normal groups can be written as
\small
\begin{align}
l_{\text{LN-CPC}}\left(\bPsi\right) = &
\ C -\frac{1}{2} \sum_{j=1}^k n_j\left[ \sum_{h=1}^d \ln \lambda_{jh} + \mbox{tr}\left(\bQ \bLambda_j^{-1}\bQ' \bS_j\right)\right] + \nonumber \\
&  
+ \sum_{j=1}^k\sum_{i=1}^{n_j} 
\ln\left\{1 +
\frac{\beta_j}{8d(d+2)}
\left[
\mbox{tr}\left(\bQ \bLambda_j^{-1}\bQ' \bS_{ij}\right)^2 -2(d+2)\mbox{tr}\left(\bQ \bLambda_j^{-1}\bQ'\bS_{ij}\right)+d(d+2)
\right]
\right\},
\label{eq:LN-CPC loglik} 
\end{align} \normalsize
\hspace{-0.46cm} where $\bPsi=\left\{\bmu_j,\bQ,\bLambda_j,\beta_j;\text{ $j=1,\ldots,k$}\right\}$ is the whole set of parameters of cardinality $m_{\text{LN-CPC}} = m_{\text{N-CPC}} + k$.

\subsection{Computational details}
\label{subsec:Computational details}

The maximization of $l_{\text{LN-CPC}}\left(\bPsi\right)$ with respect to $\bPsi$ does not admit a closed-form solution \citep[see][for the case $k=1$]{Bagn:Punz:Zoia:Them:2016}. 
Furthermore, the maximization problem is constrained due to $\bQ$, $\bLambda_j$, and $\beta_j$, $j=1,\ldots,k$.
Even if we marginalize the objective function with respect to $\beta_1,\ldots,\beta_k$, the FG-algorithm, which does not depend on the assumption of multivariate normality in the $k$ subpopulations \citep[][p.~71 and Section~9.3]{Flur:comm:1988}, can not be used to find the ML estimate of $\bQ$ and $\bLambda_1,\ldots,\bLambda_k$.
All these arguments give us the opportunity to appreciate the advantages of the proposed PLR decomposition.

To make the maximization of $l_{\text{LN-CPC}}\left(\bPsi\right)$ unconstrained, we follow a transformation/back-transformation approach from the original constrained parameters to unconstrained real values.
The constrained orthogonal matrix $\bQ$ is mapped to a $\left(d\times d\right)$ unit lower triangular matrix $\bL$, having $d\left(d-1\right)/2$ unconstrained real valued entries, via the PLR decomposition
\begin{equation}
\bQ = \bP \bL \bR^{-1},
\label{eq:Gamma transformation}
\end{equation}
where we recall that $\bP$ is a $\left(d\times d\right)$ permutation matrix and $\bR$ is a $\left(d\times d\right)$ upper triangular matrix depending on $\bL$ (see Definition~\ref{theo:1dec}).
The back-transformation is
\begin{equation}
 \bL = \bP' \bQ \bR,
\label{eq:Gamma back-transformation}
\end{equation}
which can be easily obtained by starting from the QR decomposition of $\bP \bL$.
The simple \textsf{R} code \citep{R:2018} to compute \eqref{eq:Gamma transformation} and \eqref{eq:Gamma back-transformation} is given in Appendix~\ref{app:PLR decomposition in R}.
As concerns the generic diagonal element $\lambda_{jh}$ of $\bLambda_j$, $j=1,\ldots,k$ and $h=1,\ldots,d$, the transformation is
\begin{equation}
\lambda_{jh}=\exp\left(\widetilde{\lambda}_{jh}\right),
\label{eq:lambda transformation}
\end{equation}
with $\widetilde{\lambda}_{jh} \in \real$, while the back-transformation is
\begin{equation}
\widetilde{\lambda}_{jh} = \log\left(\lambda_{jh}\right).
\label{eq:lambda back-transformation}
\end{equation}
Finally, regarding $\beta_j$, $j=1,\ldots,k$, the transformation is
\begin{equation}
\beta_j=\beta_{\max}\frac{\exp\left(\widetilde{\beta}_j\right)}{1+\exp\left(\widetilde{\beta}_j\right)},
\label{eq:beta transformation}
\end{equation}
with $\widetilde{\beta}_j \in \real$, while the back-transformation is
\begin{equation}
\widetilde{\beta}_j = \log\left(\frac{\beta_j}{\beta_{\max}-\beta_j}\right).
\label{eq:beta back-transformation}
\end{equation}
Based on \eqref{eq:Gamma transformation}, \eqref{eq:lambda transformation} and \eqref{eq:beta transformation}, in the transformation step of our procedure we maximize the log-likelihood function $l_{\text{LN-CPC}}$ with respect to $\bmu$ (which does not require any transformation), $\widetilde{\lambda}_{jh}$, $\widetilde{\beta}_j$, and to the elements below the diagonal of $\bL$, $j=1,\ldots,k$ and $h=1,\ldots,d$.
Operationally, we perform this unconstrained maximization via the general-purpose optimizer \texttt{optim()} for \textsf{R}, included in the \textbf{stats} package.
We try two different algorithms (Nelder-Mead and BFGS) for maximization.
They can be passed to \texttt{optim()} via the argument \texttt{method}.
In the back-transformation step of our procedure, the values of $\bL$, $\widetilde{\lambda}_{jh}$, and $\widetilde{\beta}_j$ maximizing $l_{\text{LN-CPC}}$ can be simply inserted in \eqref{eq:Gamma transformation}, \eqref{eq:lambda transformation} and \eqref{eq:beta transformation}, respectively, in order to obtain the ML estimates of $\bQ$, $\lambda_{jh}$, and $\beta_j$, $j=1,\ldots,k$ and $h=1,\ldots,d$.

Initial (real) values are required by \texttt{optim()} for maximization.
We use the group-specific sample mean vectors for $\bmu_1,\ldots,\bmu_k$.
For $\bQ$ and $\bLambda_1,\ldots,\bLambda_k$ we adopt the the simple intuitive procedure proposed by \citet{Krza:Prin:1984}, which is based on the PCA of the pooled sample covariance matrix and the total sample covariance matrix, followed by comparison of their eigenvectors.
Finally, to initialize $\beta_j$, $j=1,\ldots,k$, we use the empirical excess kurtosis if it falls in $\left[0,\beta_{\max}\right]$ (cf.~Section~\ref{subsec:The model}); instead, we put $\beta_j=0$, or $\beta_j=\beta_{\max}$, if the empirical excess kurtosis is lower than 0, or greater than $\beta_{\max}$, respectively.
By means of \eqref{eq:Gamma back-transformation}, \eqref{eq:lambda back-transformation} and \eqref{eq:beta back-transformation}, the obtained initial values are transformed so to be passed to \texttt{optim()}.
From the transformation \eqref{eq:Gamma transformation} related to the initial orthogonal matrix $\bQ$, we also obtain the permutation matrix $\bP$ that will be used by \texttt{optim()}.
Note that, fixing the permutation matrix to the initial one does not reduce the space of orthogonal matrices considered by \texttt{optim()}, but simply affects the order of the eigenvalues on the diagonal of $\bLambda_j$, $j=1,\ldots,k$.
This means that the ML estimated eigenvalues may not be simultaneously ordered in decreasing order in all groups.
However, having eigenvalues in an arbitrary order is not an issue in CPCA \citep{Trend:Step:2010}.

\section{Application to allometric studies}
\label{sec:Application to allometric studies}

Allometric studies are a natural area of application of the leptokurtic-normal CPCA proposed in Section~\ref{sec:Application}.
Allometry can be roughly devised as a tool to study the relation between parts (morphometric variables) in various organisms \citep{Huxl:Prob:1932}.
According to \citet{Joli:Them:1963}, allometry can be summarized by the first principal component (PC) of the log-transformed measurements.
For the practical and theoretical reasons why it is often useful to transform data to logarithms see, e.g., \citet{Pime:Morp:1979}, \citet{Reym:Mult:1991}, and \citet{Book:Morp:1997}.

When the study involves several groups of specimens, e.g.~different sexes or species, the interest is comparing the group-specific allometric patterns.
This aim can be handled by comparing the group-specific PCs \citep[see, e.g.,][]{Klin:Mult:1996}.
Comparisons of allometry within several groups often show that the PCs differ only minimally. 
In these cases, it may be feasible to use CPCA, where the groups are assumed to share a common allometric pattern, i.e., that the major axes of their scatters are parallel \citep{Klin:Neue:Flur:Onto:1996}.
The amount of variation associated with each PC can, instead, vary between groups.
However, classical CPCA implies groups having a multivariate normal distribution, and this could be rather restrictive in some cases (cf.~Section~\ref{subsec:The model}).

Motivated by the above considerations, and using classical real biometric data, we compare: the CPC model based on normal groups (N-CPC), the CPC model based on leptokurtic-normal groups (LN-CPC), the model with unconstrained covariance matrices and normal groups (N-PC), and the model with unconstrained covariance matrices and leptokurtic-normal groups (LN-PC).
The whole analysis is conducted in \textsf{R}.
Parameters of the competing models are estimated by the ML approach.
For uniformity sake, the PLR decomposition is adopted for both the CPC approaches.
For the N-CPC model, the transformation/back-transformation approach based on the PLR decomposition provided the same estimates of $\bQ$ of the FG-algorithm in all our analyses (also those not reported in this paper).
Having the competing models a differing number of parameters, their comparison is accomplished, as usual, via the Akaike information criterion (AIC; \citealp{Akai:Anew:1974}) and the Bayesian information criterion (BIC; \citealp{Schw:Esti:1978}) that, in our formulation, need to be maximized.
Moreover, likelihood-ratio (LR) tests are used to compare nested models, and this gives rise to new testing procedures.
Just as an example, the LR test having the N-CPC model under the null and the LN-PC model as alternative represents a more \textit{omnibus} version of the LR test proposed by \citet{Flur:comm:1984} where a more restrictive N-PC model is considered under the alternative. 
Using Wilks' theorem, the LR statistic, under the null, is distributed approximately as a $\chi^2$ with degrees of freedom given by the difference in the number of estimated parameters between the alternative and the null model; this allows us to compute a $p$-value that, in the sequel, will be always compared to the classical 5\% significance level.

\subsection{Skull dimensions of voles}
\label{subsec:Skull dimensions of voles}

The first analysis considers the \texttt{microtus} data set accompanying the \textbf{Flury} package \citep{R:Flury} for \textsf{R}.
The data set contains morphological measurements, for eight variables, on the skulls of 288 specimens of voles found at various places in central Europe.
For 89 of the skulls, the chromosomes were analyzed to identify their membership to one of $k=2$ species; $n_1 = 43$ specimens were from \textit{microtus multiplex}, and $n_2 = 46$ from \textit{microtus subterraneus}.
Species was not determined for the remaining 199 specimens.
\citet{Airo:Flur:Salv:Disc:1995} report a discriminant analysis and finite mixture analysis of this data set; see also \citet[][Examples~5.4.4 and 9.5.1]{Flur:Afir:2013}.

We analyze the sample of $n=89$ labeled skulls -- because we need to know the group of membership of the observations for the application of the competing models -- and focus the attention on the logarithm of $d=2$ of the available measurements: the length of palatal bone (in mm/1000) and the skull width across rostrum (in mm/100). 
The scatter plot of the observations, with symbols diversified with respect to the species, is displayed in \figurename~\ref{fig:microtus}.
\begin{figure}[!ht]
	\centering		
\resizebox{0.5\textwidth}{!}{\includegraphics{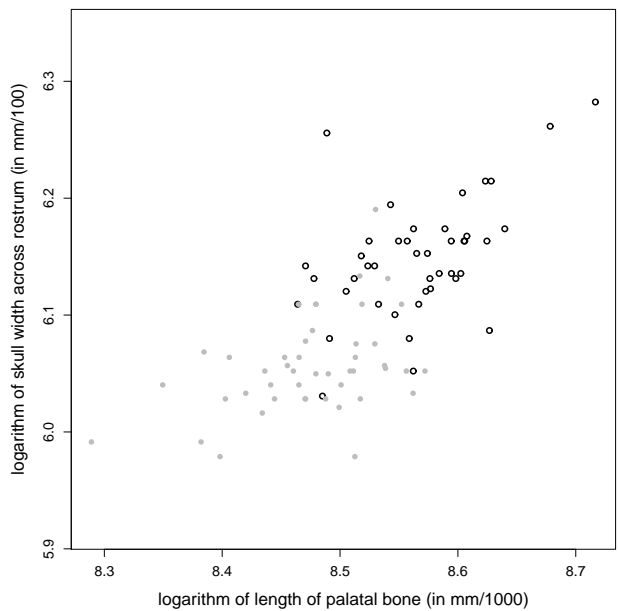}}
	\caption{
	Scatter plot of the \texttt{microtus} data with variables in logarithmic scale
	(\textcolor{gray}{$\bullet$} denotes \textsl{microtus subterraneus}, while $\boldsymbol{\circ}$ denotes \textit{microtus multiplex}).
	}
	\label{fig:microtus}
\end{figure}
For \textit{microtus multiplex} voles, the empirical excess kurtosis is 2.341, and the Mardia test of mesokurtosis, as implemented by the \texttt{mvn()} function of the \textbf{MVN} package \citep{R:MVN}, provides a $p$-value of 0.019; this leads to a rejection of the null hypothesis of mesokurtosis.
This also implies a rejection of bivariate normality.
Instead, for \textit{microtus subterraneus} voles, the empirical bivariate excess kurtosis is 1.370, corresponding to a $p$-value of 0.170 for the Mardia test. 
So, the \textit{microtus multiplex} voles motivate the need of a distribution accounting for heavier than normal tails. 
Moreover, \figurename~\ref{fig:microtus} seems to suggest that the two scatters for \textsl{microtus subterraneus} and \textit{microtus multiplex}
have approximately the same orientation (i.e., the same principal components).
 
To evaluate if a leptokurtic-normal distribution fits the data in each group better, and to assess our conjecture about the similarity between orientations, we proceed with the fitting of the competing models.
We could have used the LN distribution only for the \textit{microtus multiplex} voles, and the fitting procedure outlined in Section~\ref{subsec:Computational details} would have been easily adapted to the case.
However, this would have gone beyond the scope of this real data application, which is to show the versatility of our PLR decomposition, jointly with any unconstrained optimization routine, in the estimation of an orthogonal matrix, regardless of the model structure.
\tablename~\ref{tab:microtus estimates} reports the ML-estimated parameters. 
\begin{table}[!ht]
\caption{ML estimated parameters of the competing models for the \texttt{microtus} data. Eigenvalues are multiplied by 100.}
\label{tab:microtus estimates}
\centering
\resizebox{\textwidth}{!}{
\begin{tabular}{cccc}
\toprule
N-CPC & LN-CPC & N-PC & LN-PC \\
\midrule
$\bmu_1=\begin{bmatrix}
8.566 \\ 
6.149 \\
\end{bmatrix}$
& 
$\bmu_1=\begin{bmatrix}
8.563 \\ 
6.148 \\
\end{bmatrix}$
& 
$\bmu_1=\begin{bmatrix}
8.566 \\ 
6.149 \\
\end{bmatrix}$
& 
$\bmu_1=\begin{bmatrix}
8.563 \\ 
6.148 \\
\end{bmatrix}$ \\
$\bmu_2=\begin{bmatrix}
8.476 \\ 
6.056 \\
\end{bmatrix}$
& 
$\bmu_2=\begin{bmatrix}
8.478 \\ 
6.056 \\
\end{bmatrix}$
& 
$\bmu_2=\begin{bmatrix}
8.476 \\ 
6.056 \\
\end{bmatrix}$
& 
$\bmu_2=\begin{bmatrix}
8.478 \\ 
6.056 \\
\end{bmatrix}$ \\
\multirow{2}{*}{
$\bQ=\begin{bmatrix}
-0.850 & -0.527 \\
-0.527 &  0.850 \\
\end{bmatrix}$}
&
\multirow{2}{*}{
$\bQ=\begin{bmatrix}
-0.886 & -0.463 \\
-0.463 &  0.886 \\
\end{bmatrix}$}
&
$\bQ_1=\begin{bmatrix}
-0.766 & -0.643 \\
-0.643 &  0.766 \\
\end{bmatrix}$
&
$\bQ_1=\begin{bmatrix}
-0.851 & -0.525 \\
-0.525 &  0.851 \\
\end{bmatrix}$
\\
& 
&
$\bQ_2=\begin{bmatrix}
-0.908 & -0.420 \\
-0.420 &  0.908 \\
\end{bmatrix}$
&
$\bQ_2=\begin{bmatrix}
-0.918 & -0.395 \\
-0.395 &  0.918 \\
\end{bmatrix}$ 
\\
$\bLambda_1=\begin{bmatrix}
  0.408 & 0 \\ 
 0 & 0.140\\
\end{bmatrix}$
& 
$\bLambda_1=\begin{bmatrix}
  0.437 & 0 \\ 
 0 & 0.131\\
\end{bmatrix}$
& 
$\bLambda_1=\begin{bmatrix}
 0.413 & 0 \\ 
 0 & 0.134\\
\end{bmatrix}$
& 
$\bLambda_1=\begin{bmatrix}
  0.432 & 0 \\ 
 0 & 0.130\\
\end{bmatrix}$ \\
$\bLambda_2=\begin{bmatrix}
  0.378 & 0 \\ 
 0 & 0.123\\
\end{bmatrix}$
& 
$\bLambda_2=\begin{bmatrix}
 0.376  & 0 \\ 
 0 & 0.121 \\
\end{bmatrix}$
& 
$\bLambda_2=\begin{bmatrix}
 0.382  & 0 \\ 
 0 & 0.119\\
\end{bmatrix}$
& 
$\bLambda_2=\begin{bmatrix}
  0.378 & 0 \\ 
 0 & 0.122\\
\end{bmatrix}$ \\
&
\multirow{2}{*}{$\beta_1=  4.753    $}
&
&
\multirow{2}{*}{$\beta_1=   4.260   $}\\ 
&
\multirow{2}{*}{$\beta_2=  1.778   $}
&
&
\multirow{2}{*}{$\beta_2=  1.870   $}
\\[2mm]
\bottomrule
\end{tabular}
}
\end{table}
As we can note, the models are very similar in terms of estimated mean vectors.
Instead, they differ in terms of estimated eigenvectors and eigenvalues matrices (compare N-CPC with LN-CPC and N-PC with LN-PC).
This shows how the leptokurtik-normal distributional assumption robustifies the estimation of these matrices with respect to normally distributed groups.
The robustifying effect can be particularly noted comparing N-PC with LN-PC models; here, the estimates of the orthogonal matrices differ mainly in group 1 (\textit{microtus multiplex}), where a larger empirical/estimated excess kurtosis is present.
Finally, as concerns the models based on the leptokurtic-normal distribution, the estimates of the excess kurtoses $\beta_1$ and $\beta_2$ are in line with the empirical excess kurtoses (2.341 and 1.370).

\tablename~\ref{tab:microtus model comparison} presents a model comparison in terms of: number of parameters ($m$), log-likelihood, AIC, and BIC values (\tablename~\ref{tab:microtus loglik, AIC and BIC}) and with respect to the $p$-values from the LR tests (\tablename~\ref{tab:microtus LR tests}).
\begin{table}[!ht]
\caption{Model comparison, for the \texttt{microtus} data, in terms of number of parameters ($m$), log-likelihood, AIC, and BIC (on the left) and $p$-values from the LR tests (on the right).}
\label{tab:microtus model comparison}
\centering
\subtable[Number of parameters, log-likelihood, AIC, and BIC]
{
\label{tab:microtus loglik, AIC and BIC}
\resizebox{!}{0.055\textheight}{
\begin{tabular}{l rrrr}
\toprule
& $m$  & Log-likelihood & AIC & BIC \\
\midrule
N-CPC     & 9  &  289.319 & 560.639 & 538.241 \\
LN-CPC    & 11 &  293.946 & 565.891 & 538.516 \\
N-PC      & 10 &  290.416 & 560.831 & 535.945 \\
LN-PC     & 12 &  294.220 & 564.439 & 534.575 \\
\bottomrule
\end{tabular}
}
}
\quad
\subtable[$p$-values from the LR tests]
{
\label{tab:microtus LR tests}
\resizebox{!}{0.055\textheight}{
\begin{tabular}{l|rrr}
\toprule
\backslashbox{$H_0$}{$H_1$}  &   LN-CPC  &   N-PC   &   LN-PC \\
						\midrule
N-CPC  &      0.010 & 0.139  & 0.020 \\
LN-CPC &        --  &     -- & 0.459 \\
N-PC   &        --  &     -- & 0.022 \\
\bottomrule
\end{tabular}
}
}
\end{table}
AIC and BIC in \tablename~\ref{tab:microtus loglik, AIC and BIC} indicate LN-CPC as the best model.
This means that, with respect to a model with unconstrained covariance matrices, a more parsimonious model allowing for common principal components is to be preferred, but without giving up to groups having a heavier-tailed distribution (the leptokurtic-normal in this case). 
By looking at \tablename~\ref{tab:microtus LR tests}, the null N-CPC model of the LR test by \citet{Flur:comm:1984} is not rejected ($p$-value=0.139).
This happens because of the alternative N-PC model used by the test.
If we make the test more \textit{omnibus} by considering the LN-PC model as alternative, then the conclusion changes ($p$-value=0.020); interestingly, if we change the null hypothesis of this test using a less constrained LN-CPC model, then the conclusion is in favor of the null ($p$-value=0.459).
The null N-CPC model is also rejected, with a greater extent, if we use LN-CPC as alternative ($p$-value=0.010).
Finally, it is also interesting to note that, even if we do not consider a CPC approach, the need for leptokurtic-normal groups arises: the $p$-value of the test of N-PC versus LN-PC is, indeed, 0.022.


\subsection{Head dimensions of young Swiss soldiers}
\label{subsec:Head dimensions of young Swiss soldiers}

The second analysis considers the \texttt{swiss.soldiers} data set considered by \citet[][Example~2.3]{Flur:comm:1988}.
The data were collected by a group of anthropologists on approximately 900 Swiss soldiers, most of them recruits,
subdivided in $k=2$ groups according to the gender.
All the soldiers were 20 years old at the time of investigation and 25 variables were measured on their heads.
The purpose of the study, promoted by the the Swiss government in the mid-1980s, was to provide sufficient data to construct new protection masks for the members of the Swiss army.

We start by the subset of the \texttt{swiss.soldiers} data which can be obtained by merging the \texttt{swiss.head} (referred to men) and \texttt{f.swiss.head} (referred to women) data sets included in the \textbf{Flury} package.
The merged data contain the values of 6 head measurements for a subsample of $n=259$ soldiers, with $n_1 = 59$ women and $n_2=200$ men.
A detailed analysis for the men has been given by \citet[][Example~10.2]{Flur:Mult:2011} and \citet[][Example~1.2]{Flur:Afir:2013}.
In particular, we focus on the logarithm of $d=3$ of the available head measurements: the minimal frontal breadth (MFB), the true facial height (TFH), and the length from tragion to gnathion (LTG).
All measurements are in millimeters.
The matrix of pairwise scatter plots, with symbols diversified with respect to the gender, is displayed in \figurename~\ref{fig:swiss.soldiers}.
\begin{figure}[!ht]
	\centering		
\resizebox{0.65\textwidth}{!}{\includegraphics{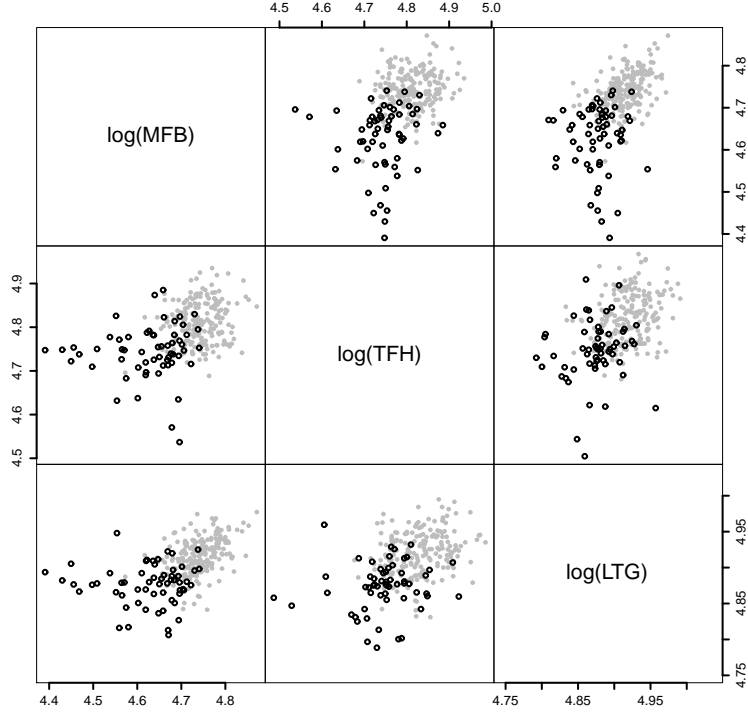}}
	\caption{
	Matrix of scatter plots of the \texttt{swiss.soldiers} data with variables in logarithmic scale ($\boldsymbol{\circ}$ denotes women and \textcolor{gray}{$\bullet$} denotes men).
	}
	\label{fig:swiss.soldiers}
\end{figure}
For women, the empirical excess kurtosis is 1.130, and the Mardia test of mesokurtosis provides a $p$-value of 0.259.
As concerns men, the empirical excess kurtosis is 7.621 and the Mardia test of mesokurtosis provides a practically null $p$-value indicating that a leptokurtic distribution should be better than the normal for that group.
So, in this example, the group of men justifies the use of the leptokurtic-normal distribution.

\tablename~\ref{tab:swiss.soldiers estimates} shows the ML estimated parameters for the competing models.
\begin{table}[!ht]
\caption{ML estimated parameters of the competing models for the \texttt{swiss.soldiers} data. Eigenvalues are multiplied by 100.}
\label{tab:swiss.soldiers estimates}
\centering
\resizebox{\textwidth}{!}{
\begin{tabular}{cccc}
\toprule
N-CPC & LN-CPC & N-PC & LN-PC \\
\midrule
$\bmu_1=\begin{bmatrix}
4.721 \\ 
4.810 \\ 
4.931 \\ 
\end{bmatrix}$
& 
$\bmu_1=\begin{bmatrix}
4.718 \\
4.807 \\ 
4.929 \\ 
\end{bmatrix}$
& 
$\bmu_1=\begin{bmatrix}
4.721 \\ 
4.810 \\ 
4.931 \\ 
\end{bmatrix}$
& 
$\bmu_1=\begin{bmatrix}
4.719 \\ 
4.808 \\ 
4.931 \\ 
\end{bmatrix}$ \\
$\bmu_2=\begin{bmatrix}
4.714 \\ 
4.792 \\
4.913 \\
\end{bmatrix}$
& 
$\bmu_2=\begin{bmatrix}
4.722 \\ 
4.796 \\
4.917 \\
\end{bmatrix}$
& 
$\bmu_2=\begin{bmatrix}
4.714 \\ 
4.792 \\
4.913 \\
\end{bmatrix}$
& 
$\bmu_2=\begin{bmatrix}
4.721 \\ 
4.795 \\
4.916 \\
\end{bmatrix}$ \\
\multirow{2}{*}{
$\bQ=\begin{bmatrix}
-0.783 &   0.546 &  -0.298 \\
-0.464 &  -0.832 &  -0.304 \\
-0.414 &  -0.100 &   0.905 \\
\end{bmatrix}$}
&
\multirow{2}{*}{
$\bQ=\begin{bmatrix}
-0.750 &   0.559 &  -0.353 \\
-0.501 &  -0.829 &  -0.247 \\
-0.431 &  -0.009 &   0.902 \\
\end{bmatrix}$}
&
$\bQ_1=\begin{bmatrix}
-0.647 &    0.230 &   -0.727 \\
-0.574 &   -0.775 &    0.266 \\
-0.502 &    0.589 &    0.633 \\
\end{bmatrix}$
&
$\bQ_1=\begin{bmatrix}
 -0.634 &  0.290 &  -0.717 \\
 -0.566 &  -0.806 &  0.174 \\
 -0.527 &  0.516 &  0.675 \\
\end{bmatrix}$
\\
& 
&
$\bQ_2=\begin{bmatrix}
-0.796 &  0.531 & -0.291 \\
-0.448 & -0.840 & -0.305 \\
-0.407 & -0.113 &  0.907 \\
\end{bmatrix}$
&
$\bQ_2=\begin{bmatrix}
-0.765 &  0.554 &  -0.327  \\
-0.486 &  -0.831 &  -0.271  \\
-0.422 &  -0.049 &  0.905  \\
\end{bmatrix}$ 
\\
$\bLambda_1=\begin{bmatrix}
 0.282  &    0   &  0  \\ 
 0 &   0.135     &  0  \\
 0 &    0   &   0.130  \\
\end{bmatrix}$
& 
$\bLambda_1=\begin{bmatrix}
 0.303  &    0   &  0  \\ 
 0 &  0.149      &  0  \\
 0 &    0   &  0.128   \\
\end{bmatrix}$
& 
$\bLambda_1=\begin{bmatrix}
 0.289  &    0   &  0  \\ 
 0 &  0.153      &  0  \\
 0 &    0   & 0.105    \\
\end{bmatrix}$
& 
$\bLambda_1=\begin{bmatrix}
 0.304  &    0   &  0  \\ 
 0 &  0.161      &  0  \\
 0 &    0   &   0.110  \\
\end{bmatrix}$ \\
$\bLambda_2=\begin{bmatrix}
  0.867 &    0   &  0  \\ 
 0 &   0.269     &  0  \\
 0 &    0   &  0.153   \\
\end{bmatrix}$
& 
$\bLambda_2=\begin{bmatrix}
  0.925 &    0   &  0  \\ 
 0 &  0.303      &  0  \\
 0 &    0   &  0.166   \\
\end{bmatrix}$
& 
$\bLambda_2=\begin{bmatrix}
  0.868 &    0   &  0  \\ 
 0 &   0.269     &  0  \\
 0 &    0   & 0.153    \\
\end{bmatrix}$
& 
$\bLambda_2=\begin{bmatrix}
 0.916  &    0   &  0  \\ 
 0 &  0.295      &  0  \\
 0 &    0   &   0.165  \\
\end{bmatrix}$ \\
&
\multirow{2}{*}{$\beta_1=  1.813   $}
&
&
\multirow{2}{*}{$\beta_1=  1.778   $}\\ 
&
\multirow{2}{*}{$\beta_2=  5.601   $}
&
&
\multirow{2}{*}{$\beta_2=  5.331   $}
\\[2mm]
\bottomrule
\end{tabular}
}
\end{table} 
As for the example on the \texttt{microtus} data of Section~\ref{subsec:Skull dimensions of voles}, the models behave in a similar way in terms of estimated mean vectors, but differ in terms of estimated eigenvectors and eigenvalues matrices if N-CPC is compared with LN-CPC and N-PC with LN-PC.
Such a difference is due to the leptokurtic-normal distributional assumption.
For the LN-based models, the estimates of the excess kurtoses $\beta_1$ and $\beta_2$ are in line enough with the empirical excess kurtoses (1.130 and 7.621).

\tablename~\ref{tab:swiss.soldiers model comparison} presents the model comparison.
\begin{table}[!ht]
\caption{Model comparison, for the \texttt{swiss.soldiers} data, in terms of number of parameters ($m$), log-likelihood, AIC, and BIC (on the left) and $p$-values from the LR tests (on the right).}
\label{tab:swiss.soldiers model comparison}
\centering
\subtable[Number of parameters, log-likelihood, AIC, and BIC]
{
\label{tab:swiss.soldiers loglik, AIC and BIC}
\resizebox{!}{0.055\textheight}{
\begin{tabular}{l rrrr}
\toprule
           & $m$ par. & Log-likelihood & AIC & BIC \\
\midrule
N-CPC     & 15 &  1176.166 & 2322.333 & 2268.980 \\
LN-CPC    & 17 &  1190.986 & 2347.973 & 2287.506 \\
N-PC      & 19 &  1178.091 & 2320.182 & 2256.159 \\
LN-PC     & 20 &  1192.435 & 2344.869 & 2273.733 \\
\bottomrule
\end{tabular}
}
}
\quad
\subtable[$p$-values from the LR tests]
{
\label{tab:swiss.soldiers LR tests}
\resizebox{!}{0.055\textheight}{
\begin{tabular}{l|rrr}
\toprule
\backslashbox{$H_0$}{$H_1$} &   LN-CPC  &   N-PC   &   LN-PC \\
						\midrule
N-CPC  &      0.000 & 0.278  & 0.000 \\
LN-CPC &        --  &     -- & 0.408 \\
N-PC   &        --  &     -- & 0.000 \\
\bottomrule
\end{tabular}
}
}
\end{table}
AIC and BIC in \tablename~\ref{tab:swiss.soldiers loglik, AIC and BIC} select LN-CPC as the best model, with a stronger extent with respect to the previous example.
Such a greater extent has analogous implications for the $p$-values from the LR tests (see \tablename~\ref{tab:swiss.soldiers LR tests}).
The null N-CPC model is not rejected if N-PC is considered as alternative model ($p$-value=0.278), but it is strongly rejected (with an approximately null $p$-value) if a LN-based model (LN-CPC or LN-PC) is considered under the alternative.
If LN-CPC is considered as the null model, and LN-PC as alternative, then the conclusion is in favor of the null ($p$-value=0.408).
Even if we do not consider a CPC approach, the need for leptokurtic-normal groups is even more evident in this example: the $p$-value of the test of N-PC versus LN-PC is, indeed, practically null.

\section{Conclusions}
\label{sec:Conclusions}

Estimating and modelling a $\left(d \times d\right)$ covariance matrix $\bSigma$ is often difficult because of the constraint that $\bSigma$ must be positive definite.
The Cholesky decomposition provides a remedy to this problem by mapping the $d\left(d+1\right)/2$ constrained parameters of $\bSigma$ with the $d\left(d+1\right)/2$ unconstrained real-valued elements of a $\left(d \times d\right)$ lower triangular matrix (\citealp{pourahmadi99,pourahmadi00} and \citealp{pourahmadi07}).
Analogously, estimating and modelling a $\left(d \times d\right)$ orthogonal matrix $\bQ$ is often cumbersome because of its orthonormality constraints.
Unfortunately, in this case, an analogous of the Cholesky decomposition does not exist.
In this paper we have filled this gap by providing a decomposition for orthogonal matrices, that we have called PLR decomposition, mapping $\bQ$ to a $\left(d \times d\right)$ unit lower triangular matrix with $d\left(d-1\right)/2$ unconstrained real entries below the diagonal.  

For illustrative purposes, we have applied our PLR decomposition in common principal component analysis (CPCA), based on groups having a heavier-tails leptokurtic-normal distribution, for maximum likelihood estimation of the common orthogonal matrix.
We have chosen allometry as a natural area of application of the resulting leptokurtic-normal CPCA and the real data analyses have effectively confirmed its good behavior when compared to the classical normal CPCA.

However, the use of the PLR decomposition of an orthogonal matrix is not limited to CPCA, and other statistical models may benefit from its use.
Indeed, the PLR decomposition may be used to simplify the ML estimation of the orthogonal matrix related, only to mention a few, to: 
CPCA based on further non-normal distributions for the groups,  
other multiple group models allowing for common covariance structures (\citealp{Flur:Prop:1986} and \citealp{Gres:Punz:Clos:2013}),    
parsimonious model-based clustering, classification and discriminant analysis (\citealp{Banf:Raft:mode:1993}, \citealp{Flur:Schm:Nara:Erro:1994}, \citealp{Cele:Gova:Gaus:1995}, \citealp{Fral:Raft:Mode:2002}, \citealp{Andr:McNi:Mode:2012}, \citealp{Bagn:Gres:Punz:Onth:2014}, \citealp{Lin:Lear:2014}, \citealp{Vrbi:McNi:2014}, \citealp{Dang:Brow:McNi:Mixt:2015}, \citealp{Punz:Brow:McNi:Hypo:2016}, \citealp{Punz:McNi:Robu:2016}, \citealp{R:ContaminatedMixt:JSS}, \citealp{Dott:Farc:Robu:2019}), extensions of hidden Markov models (\citealp{Punz:Maru:Clus:2016} and \citealp{Maru:Punz:Mode:2017}),  
and sophisticated multivariate distributions (\citealp{Forb:Wrai:Anew:2014} and \citealp{Punz:Tort:Mult:2018}).
We pursue to handle these possibilities in future works.


\clearpage

\appendix


\section{}
\subsection{PLR decomposition in \textsf{R}}
\label{app:PLR decomposition in R}


\begin{lstlisting}
library(Matrix)

#--------------#
# From Q to PL #
#--------------#

Q.PL <- function(Q){
  
  PLR  <- Matrix::lu(Q)
  ePLR <- Matrix::expand(PLR)
  
  return(list(P = ePLR$P, L = ePLR$L))
  
}

#--------------#
# From PL to Q #
#--------------#

PL.Q <- function(P, L){
  
  A  <- P %*% L
  QR <- base::qr(A)
  R  <- base::qr.R(QR)
  Q  <- P %*% L %*% solve(R) 
  
  return(Q)
  
}


\end{lstlisting}

\label{lastpage}

\end{document}